\documentclass[runningheads]{llncs}
\usepackage[T1]{fontenc}
\usepackage{cite}
\usepackage{algorithm}
\usepackage{algorithmic}
\usepackage{amsmath}
\usepackage{amssymb}
\usepackage{amsfonts}
\usepackage{graphicx}
\usepackage{textcomp}
\usepackage{xcolor}
\usepackage{stfloats}
\usepackage{caption}
\usepackage{float}
\usepackage{appendix}
\usepackage{marvosym} % for \Letter symbol
\usepackage[unicode=false,pdfusetitle]{hyperref}
\begin{document}
\title{Efficient and Secure Sleepy Model for BFT Consensus %
\thanks{This paper has been accepted by the 30th European Symposium on Research in Computer Security (ESORICS 2025).}
}

\author{Pengkun Ren\inst{1} \and
Hai Dong\inst{1}\thanks{\Letter~\email{hai.dong@rmit.edu.au} (Corresponding Author)}
\and
Zahir Tari\inst{1} \and
Pengcheng Zhang\inst{2}\relax}

% \institute{ \inst{1} School of Computing Technologies, Centre of Cyber Security Research and Innovation, RMIT University, Melbourne, Australia\\
%     \email{s4038427@student.rmit.edu.au}, 
%     \email{\{hai.dong, zahir.tari\}@rmit.edu.au}\\
%   \inst{2} College of Computer Science and Software Engineering, Hohai University, China\\
%     \email{pchzhang@hhu.edu.cn}}
\institute{School of Computing Technologies, Centre of Cyber Security Research and Innovation, RMIT University, Melbourne, Australia\\
    \email{s4038427@student.rmit.edu.au},
    \email{\{hai.dong, zahir.tari\}@rmit.edu.au}
    \and % 使用 \and 来分隔不同的机构
    College of Computer Science and Software Engineering, Hohai University, China\\
    \email{pchzhang@hhu.edu.cn}}
  
% \email{lncs@springer.com}\\
% \url{http://www.springer.com/gp/computer-science/lncs} \and
% ABC Institute, Rupert-Karls-University Heidelberg, Heidelberg, Germany\\
% \email{\{abc,lncs\}@uni-heidelberg.de}}
%
\maketitle              % typeset the header of the contribution
\begin{abstract}
Byzantine Fault Tolerant (BFT) consensus protocols for dynamically available systems face a critical challenge: balancing latency and security in fluctuating node participation. Existing solutions often require multiple rounds of voting per decision, leading to high latency or limited resilience to adversarial behavior. This paper presents a BFT protocol integrating a pre-commit mechanism with publicly verifiable secret sharing (PVSS) into message transmission. By binding users' identities to their messages through PVSS, our approach reduces communication rounds. Compared to other state-of-the-art methods, our protocol typically requires only four network delays (4$\Delta$) in common scenarios while being resilient to up to 1/2 adversarial participants. This integration enhances the efficiency and security of the protocol without compromising integrity. Theoretical analysis demonstrates the robustness of the protocol against Byzantine attacks. Experimental evaluations show that, compared to traditional BFT protocols, our protocol significantly prevents fork occurrences and improves chain stability. Furthermore, compared to longest-chain protocol, our protocol maintains stability and lower latency in scenarios with moderate participation fluctuations.

\end{abstract}
\section{Introduction}
Recent advancements in BFT have increasingly focused on addressing the challenges of dynamically available systems. BFT consensus protocols are fundamental to distributed systems, enabling reliable agreement among participants even in the presence of malicious behavior \cite{castro1999practical}. 

Traditional consensus protocols often assume a static set of participants who remain active throughout the entire execution of the protocol\cite{castro1999practical}\cite{katz2009expected}. The Sleepy Model, introduced by Pass and Shi \cite{pass2017sleepy}, represents a significant paradigm shift by allowing nodes to dynamically switch between active and inactive states without prior notice. This model more accurately reflects the operational realities of modern distributed systems, where node availability can fluctuate due to network instability, hardware failures, or resource optimization strategies.

Protocols based on the Sleepy Model have made notable strides in accommodating fluctuating network participation, demonstrating practical resilience and efficient latency management under dynamic conditions \cite{goyal2021instant}\cite{d2022goldfish}\cite{momose2022constant}\cite{malkhi2023towards}\cite{d2023streamlining}\cite{d2024recent}\cite{wang2024sleepy}. However, existing sleepy BFT models face a fundamental challenge in balancing latency and security. This latency-security trade-off manifests itself primarily in two critical aspects:

\textbf{1. Multiple Round Overhead:} Many existing protocols require multiple rounds of communication to reach consensus, significantly increasing overall latency. For instance, the work of \cite{momose2022constant} requires 16$\Delta$ for consensus. Although some models introduced early decision mechanisms \cite{malkhi2023towards}\cite{d2023streamlining}, achieving 4$\Delta$ latency in the best-case scenario, the average latency remains considerably high. 
    
\textbf{2. Limited Resilience to Adversarial Behavior:} As participation fluctuates, maintaining robust security against adaptive adversaries becomes increasingly challenging. In dynamic participation scenarios, sophisticated adversarial strategies can exploit temporary imbalances in the network \cite{wang2024sleepy}, potentially compromising the integrity of the consensus process. These attacks underscore the difficulty of maintaining consistent security guarantees in systems where the set of active participants changes over time.

Our approach incorporates a pre-commit mechanism with PVSS into the consensus process. The protocol achieves a balanced latency-security trade-off, maintaining security guarantees against up to 1/2 malicious nodes while achieving a latency of 4$\Delta$ in typical scenarios. Our key contributions are as follows: 

\textbf{PVSS-based Message Binding Mechanism:} Our approach addresses these challenges by binding messages with node identities through PVSS and integrating a pre-commit mechanism. In our protocol, nodes commit to their future participation before each consensus round, with these commitments cryptographically bound to their messages through PVSS. This design enables secure message verification in a single voting round by requiring nodes to distribute verifiable shares that can only be reconstructed with sufficient participation from the committed set, effectively preventing selective message broadcasting. When nodes plan to go offline for maintenance or upgrades, they can signal this through the pre-commit mechanism, while any unplanned deviation from commitments becomes cryptographically evident through failed PVSS reconstruction. In typical scenarios, only one secret reconstruction operation is needed, specifically for the leader's proposed block, though additional reconstructions may be required in consensus failure cases. As demonstrated in Table~\ref{tab:tob-comparison}, this approach achieves competitive latency performance compared to state-of-the-art protocols.
    
\textbf{Enhanced Security Model for Dynamic Participation Environments:} We provide a rigorous mathematical model of node behavior within each $\Delta$ time interval. This model leverages PVSS properties to maintain security integrity even under sophisticated attack scenarios in dynamic settings. Additionally, the pre-commit mechanism requires that nodes must send participation commitments before joining the next consensus round, enabling dynamic threshold adjustments and enhancing the protocol's adaptability to fluctuating network conditions. We present comprehensive proofs demonstrating the protocol's resilience against a range of adversarial strategies, including adaptive corruptions and network partitioning attacks.

\begin{table*}[!t]
\centering
\caption{Comparison of dynamically sleepy models}
\label{tab:tob-comparison}
\begin{tabular}{|l|c|c|c|c|c|}
\hline
Protocol & Adv.Resilience & Best.Latency & Avg.Latency & Block Time & Voting Rounds \\
\hline
MR \cite{momose2022constant} & 1/2 & 16$\Delta$ & 24$\Delta$ & 16$\Delta$ & 10 \\
MMR2 \cite{malkhi2023towards} & 1/2 & 4$\Delta$ & 9$\Delta$ & 10$\Delta$ & 9 \\
GL \cite{gafni2023brief} & 1/2 & 6$\Delta$ & 11$\Delta$ & 10$\Delta$ & 9 \\
DZ \cite{d2023streamlining} & 1/2 & 6$\Delta$ & 8$\Delta$ & 4$\Delta$ & 1 \\
1/3MMR \cite{malkhi2022byzantine} & 1/3 & 3$\Delta$ & 4$\Delta$ & 2$\Delta$ & 2 \\
1/4MMR \cite{malkhi2022byzantine} & 1/4 & 2$\Delta$ & 3$\Delta$ & 2$\Delta$ & 1 \\
DNTS \cite{d2022goldfish} & 1/2 & $O(\kappa)^*$ & $O(\kappa)^*$ & 3$\Delta$ & 1 \\
PS \cite{pass2017sleepy} & 1/2 & $O(\kappa)$ & $O(\kappa)$ & $\Delta$ & 0 \\
\textbf{PVSS-BFT} & 1/2 & 4$\Delta$ & 4$\Delta$ & 4$\Delta$ & 1 \\
\hline
\end{tabular}
\par
\smallskip
\raggedright Note: We compare several sleepy models in terms of their resilience and performance. The definitions of latency, block time, and voting rounds are adopted from existing literature \cite{d2023streamlining}. "Adv. Resilience" refers to Adversarial Resilience. "Avg. Latency" refers to the expected time for a transaction's confirmation under random submission times. These metrics help in evaluating the efficiency and security of each protocol.

$^*$ For Goldfish\cite{d2022goldfish}, we show the latency in conditions of low participation.
\end{table*}

Our experimental results demonstrate the effectiveness of our protocol in dynamic participation environments. In simulated attack scenarios, our protocol significantly reduces fork occurrences and maintains chain stability compared to traditional BFT protocols. Furthermore, latency tests show that PVSS-BFT achieves constant low latency across various participation levels, outperforming longest-chain based systems in the circumstances of low and high participation with moderate fluctuations. These findings underscore the robustness and efficiency of our approach to balancing security and performance in sleepy models.

The rest of this paper is structured as follows: Section~\ref{related work} provides a detailed overview of related work, Section~\ref{model} describes our system model and definitions, Section~\ref{pvss-bft} presents our consensus protocol, Section~\ref{security} provides a comprehensive security analysis, Section~\ref{experiment} presents our experimental results, and Section~\ref{conclusion} concludes the paper with a discussion.

\section{Related Work}
\label{related work}

\textbf{Sleepy Consensus Model}
Traditional BFT protocols, such as PBFT \cite{castro1999practical}, were designed with the assumption of a static set of continuously participating nodes. However, node availability can fluctuate in real-world distributed systems. The \emph{Sleepy Consensus Model}, introduced by Pass and Shi \cite{pass2017sleepy}, addresses the challenges posed by dynamic node availability. This model allows nodes to switch between online and offline, enabling consensus protocols to function effectively even when the set of participating nodes changes over time. Daian et al.\cite{daian2019snow} extended the Sleepy Model to Proof-of-Stake (PoS) systems with their Snow White protocol. Goyal et al.\cite{goyal2021instant} proposed a method for instant confirmation of transactions after their appearance in the ledger.

Building on these foundational efforts, Momose and Ren \cite{momose2022constant} further advanced the concept by introducing the notion of Dynamic Quorum. Their protocol incorporates graded agreement into the consensus process, allowing for the adjustment of quorum sizes at any given time. This innovation paved the way for subsequent research, enabling the development of Total-Order Broadcast protocols capable of accommodating inherently fluctuating node availability. Gafni and Losa \cite{gafni2023brief} utilized a Commit-Adopt primitive within the sleepy model.  Malkhi et al. \cite{malkhi2023towards} addressed the issue of potentially corrupt nodes while maintaining the integrity of the consensus. D'Amato and Zanolini \cite{d2023streamlining} optimized the Sleepy Model by integrating an extended three-grade graded agreement into the consensus process. This extension streamlined the voting process, reducing the number of rounds required to reach consensus. Wang et al. \cite{wang2024sleepy} focused on the known participation model, where the network is aware of a minimum number of awake honest replicas. These developments have improved the ability of consensus protocols to manage dynamic participation. However, \emph{multiple rounds of communication} can increase latency, while dynamically adjusting the quorum size increases the complexity of maintaining safety and liveness.

\textbf{Security Optimizations in BFT Protocols}
Though designed to withstand arbitrary faults, BFT protocols struggle to maintain security as networks scale and adversaries grow more sophisticated. Researchers have continuously sought to enhance the security of BFT systems. Algorand \cite{chen2019algorand} introduced a verifiable random function (VRF) to randomly select a user committee to run the consensus in each round. ByzCoin \cite{kogias2016enhancing} proposed the use of collective signing to reduce communication complexity and maintain safety. SBFT\cite{gueta2019sbft} further utilized threshold signatures based on the BLS signature scheme\cite{boneh2001short}.

Secret sharing \cite{blakley1979safeguarding}\cite{shamir1979share} schemes have been increasingly integrated into BFT protocols to enhance their security \cite{beimel2011secret}\cite{chandramouli2022survey}. FastBFT\cite{liu2018scalable} leverages trusted execution environments (TEE) to design a fast and scalable BFT protocol. It uses TEE for generating and sharing secrets, and introduces a tree-based communication topology to distribute load. Basu et al.\cite{basu2019efficient} introduce an efficient VSS scheme with share recovery capabilities designed for BFT integration. COBRA\cite{vassantlal2022cobra} proposes VSSR, a VSS framework that uses hardware-assisted TEE and lightweight secret sharing to efficiently aggregate messages. These advancements demonstrate the significant potential for integrating cryptographic primitives such as VRF and secret sharing into BFT protocols to improve security. However, while secret sharing schemes have been effectively employed to improve BFT protocols, their application in environments with dynamic node participation remains limited.

\section{Model and Definitions}
\label{model}

We consider a system of $n$ validators $\mathcal{N} = \{n_1, ..., n_n\}$ participating in a BFT protocol. The protocol proceeds in a series of consecutive views, each denoted by an integer $v$. In each view, one or more blocks may be proposed, but only one block can be decided \cite{d2023streamlining}.
Previous works \cite{pass2017sleepy}\cite{goyal2021instant}\cite{momose2022constant} have demonstrated that network synchrony is a necessary condition for consensus protocols in environments where nodes can dynamically participate, as asynchronous networks cannot guarantee consensus when node participation is unpredictable. An adaptive adversary exists in the system that can corrupt nodes at any point during execution. These corrupted nodes may exhibit arbitrary Byzantine behavior, deviating from the protocol specifications and potentially colluding with other malicious nodes. Nodes that remain uncorrupted throughout execution are considered honest and follow protocol specifications faithfully.

\textbf{Weakly synchronized clocks. }
In our model, all participating nodes maintain synchronized clocks within a maximum deviation of $\Delta$ from the global reference time, where $\Delta$ denotes the maximum tolerable network delay for message delivery and clock skew. At any global time $t$, each node $p$ maintains a local time $\tau_p$ that differs from $t$ by at most $\Delta$, specifically $\tau_p = t - \delta_p$ where $0 \leq \delta_p \leq \Delta$. Following the definition given by \cite{goyal2021instant}\cite{momose2022constant}\cite{malkhi2023towards}, without loss of generality, we will consider nodes to have synchronized clocks. 

\textbf{Communication channels.}
We assume authenticated channels for message transmission. Byzantine nodes cannot modify messages or prevent message delivery between honest nodes. Consistent with the definition provided by \cite{momose2022constant}, we assume that if an honest node $p$ is awake at global time $t$, then $p$ has received all messages sent by honest nodes by global time $t - \Delta$.

\textbf{The sleepy model. }
Our protocol operates in the sleepy model introduced by Pass and Shi \cite{pass2017sleepy}. Each validator $n_i \in \mathcal{N}$ dynamically transitions between active and inactive states over time. The awake nodes fully participate in protocol execution, while the asleep nodes suspend all protocol activities. Let $\mathcal{A}(t) \subseteq \mathcal{N}$ denote the set of active validators at global time $t$, with $|\mathcal{A}(t)| = n_t$ where $0 < n_t \leq n$. In line with \cite{momose2022constant}\cite{malkhi2023towards}\cite{d2023streamlining}, we require that at any time $t$, the number of Byzantine validators in $\mathcal{A}(t)$ is bounded by $f_t$, where $f_t < n_t/2$, ensuring an honest majority among active participants.
The protocol accommodates state transitions of honest validators under adversarial control without requiring advance notification, though our pre-commit mechanism (detailed in Section~\ref{pvss-bft}) provides a framework for planned transitions.

\textbf{Atomic broadcast. }
As defined in \cite{momose2022constant}\cite{malkhi2023towards}, atomic broadcast enables participants to reach consensus on an expanding sequence of values $[B_0, B_1, B_2, \dots]$, which is commonly referred to as a log. 
Our protocol implements a variant of atomic broadcast tailored to the PVSS-BFT mechanism. 
In our model, clients broadcast transactions to all validators, and nodes independently propose blocks containing these transactions. In each view $v$, a leader is then elected based on VRF values.
This order of operations influences the liveness property. We have adapted the liveness definition inspired by \cite{wang2024sleepy} to suit the specific mechanisms of our protocol. The protocol provides the following guarantees: 

For safety: if two honest nodes decide logs $[B_0, B_1, \dots, B_j]$ and $[B'_0, B'_1, \dots, B'_{j'}]$, then $B_i = B'_i$ for all $i \leq \min(j, j')$.

For liveness: if a transaction \(\text{tx}\) is broadcast to all honest validators at time \( t \), then there exists a time \( t' \geq t \) such that all honest validators awake at time \( t' \) will eventually decide on a log containing \(\text{tx}\).

\textbf{Block. }
Blocks represent batches of transactions. Each block $B$ contains a reference to its parent block (except the genesis block $B_0$), a set of valid transactions, and cryptographic proof of validity. Two blocks conflict if one of them extends the other at the same height.  A chain is a growing sequence of blocks $[B_0, B_1, \dots, B_j]$.

\textbf{Fork. }  A fork refers to a situation in which two or more conflicting blocks (i.e., blocks at the same height that do not extend each other) are simultaneously decided or irreversibly committed by different subsets of honest validators.

\subsection{Cryptographic Primitives}
\label{cryptographic-primitives}
We assume a public-key infrastructure (PKI) with digital signatures. A message $\mu$ signed by $p_i$ is noted as $\langle\mu\rangle_i$.$H(\cdot)$ represents a collision-resistant hash function.

\textbf{Verifiable Random Functions}
A VRF allows a participant $p_i$ to generate a verifiable pseudo-random value $\rho_i$ and a proof $\pi_i$ from an input $\mu$: $(\rho_i, \pi_i) \leftarrow \text{VRF}_i(\mu)$. Anyone can verify the correctness of $\rho_i$ using $\pi_i$ and $p_i$'s public key $pk_i$, i.e. $VRF.VERIFY(pk_i, \mu, \rho_i, \pi_i)$. 

\textbf{Secret Sharing and Verifiable Secret Sharing (VSS)}
Secret sharing introduced independently by Shamir\cite{shamir1979share} and Blakley\cite{blakley1979safeguarding} allow a secret $s$ to be splitted into share distributed among multiple parties. In a $(t,n)$-threshold scheme, any $t$ out of $n$ shares can reconstruct $s$, while fewer than $t$ shares reveal nothing about $s$. VSS schemes, introduced by Chor et al.\cite{choc1985verifiable}, allow honest participants detect inconsistency if the dealer tries to distribute inconsistent shares.

\textbf{Publicly Verifiable Secret Sharing (PVSS)}
PVSS\cite{schoenmakers1999simple}\cite{stadler1996publicly} further allows anyone to verify share correctness without compromising secrecy. In our protocol, each node that proposes a block during view $v$ plays the role of “dealer”. Let $G_q$ denote a group of prime order $q$, and $g, G \in G_q$ be independently selected generators. According to \cite{schoenmakers1999simple}, the scheme operates as follows:

Initialization: Each participant $P_i$ generates a private key $x_i \in_R \mathbb{Z}_q^*$ and registers a public key $y_i = G^{x_i}$.

Share Distribution: The dealer selects a random polynomial $p(x) = \sum_{j=0}^{t-1} \alpha_j x^j$ of degree at most $t-1$, where $\alpha_j \in \mathbb{Z}_q$ and $s = \alpha_0$. The dealer publishes commitments $C_j = g^{\alpha_j}$ for $0 \leq j < t$; encrypted shares $Y_i = y_i^{p(i)}$ for $1 \leq i \leq n$ using $PVSS.SPLIT(s, n, t)$. Anyone can check that share $Y_i$ is consistent with $C_i$ by calling $PVSS.VERIFY(pk_i, C_i, Y_i)$.
    
Share Reconstruction: Any $t$ participants can reconstruct the secret by decrypting shares: $S_i = Y_i^{1/x_i} = G^{p(i)}$ and combining the shares using Lagrange interpolation: $S = \prod_{i=1}^t S_i^{\lambda_i}$, where $\lambda_i$ are the Lagrange coefficients. This process can be abbreviated as $PVSS.RECONSTRUCT(\{S_i\}_{i=1}^t)$.

The scheme guarantees, for any $1\!\le\! t\!\le\! n$:

\emph{Correctness.} If the dealer follows the protocol, $PVSS.RECONSTRUCT$ on any $t$ honest shares always outputs the original $s$.

\emph{$t$‑Privacy.} Fewer than $t$ colluding parties obtain no information about $s$

\emph{Public Verifiability.} Anyone, given $(pk_i, C_i, Y_i)$, can run $PVSS.VERIFY$ without access to secret keys. Valid shares always pass, and forging an invalid share that passes is infeasible.

\emph{$t$‑Robustness.} Once $C$ is fixed, at most one secret $s$ can be reconstructed from any set of $t$ verified shares. A malicious dealer cannot make two inconsistent secrets pass verification.

\section{PVSS-BFT Protocol Overview}
\label{pvss-bft}

Our PVSS-BFT protocol is presented in Algorithm~\ref{alg:pvss-bft}. Figure~\ref{fig:protocol_flow} shows an overview of the flow of our protocol.

\subsection{Phase 1: Block Proposal and Share Distribution}

\begin{figure*}[!t]
    \centering
    \includegraphics[width=\textwidth]{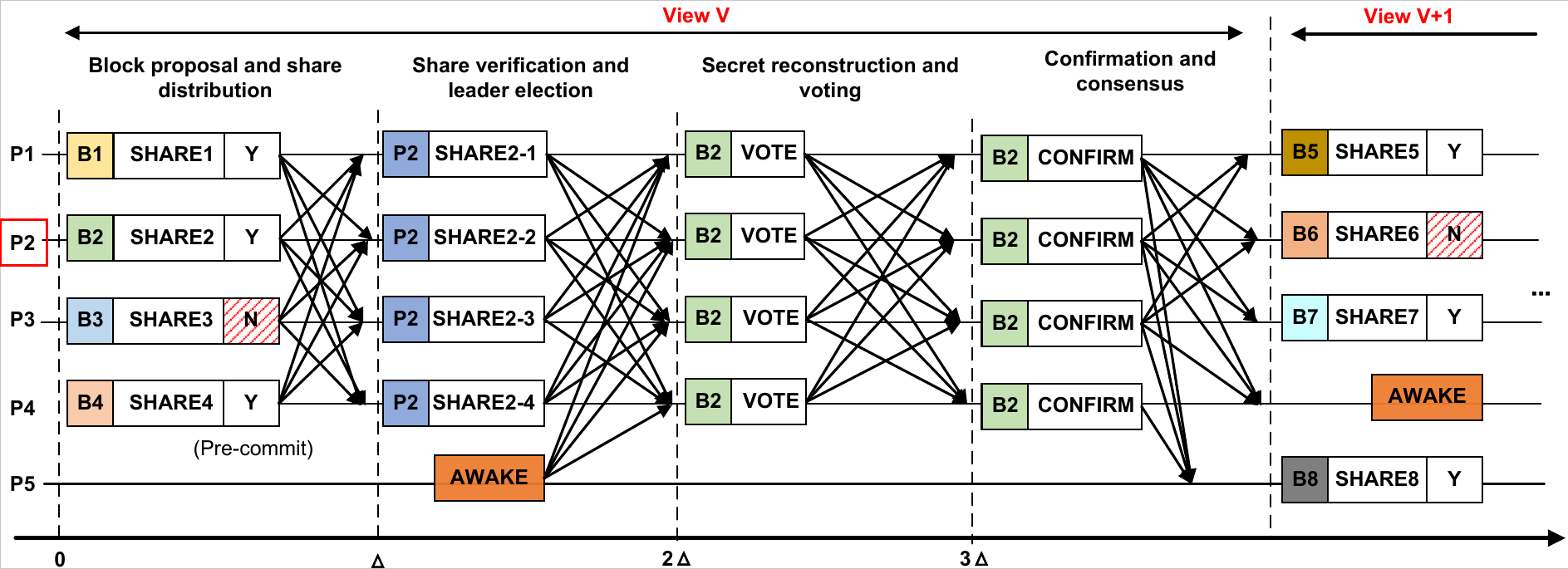}
     \caption{The PVSS-BFT protocol flow across $v$ and $v+1$. In the first phase, each node proposes a candidate block, splits the block’s hash into multiple shares, and distributes shares along with a pre-commit signal indicating whether it intends to join the next consensus round. Nodes $P1$, $P2$, and $P4$ choose to pre-commit $(Y)$, while P3 opts out $(N)$. In the second phase, nodes determine a leader (e.g. $P2$) based on VRF values. Then every node that previously received $P2$’s share now broadcasts that share to all other nodes. During phase 3, nodes use the collected shares to reconstruct the $P2$’s block and cast votes on its validity. In the final phase, upon gathering a sufficient number of valid votes, nodes broadcast confirm messages, finalizing consensus on the leader’s proposed block. As a result, by the start of View $V+1$, node P3, which did not pre-commit, is excluded, while node P5, having issued an awake message, successfully joins the protocol. }
    \label{fig:protocol_flow}
\end{figure*}

In each view $v$, we define $\mathcal{A}(v)$ from the participant set $N$ as the number of active participants recognized for $v$. $\mathcal{A}(v)$ consists of nodes that pre-committed their participation and newly awakened nodes that were confirmed in view $v-1$. Every node in $\mathcal{A}(v)$ may independently proposes a block in Phase 1 before the leader election to ensure protocol progress regardless of eventual leader availability. 

\textbf{Block Proposal and Pre-Commit. }Each node $i$ generates a block $B_i$. It also creates a pre-commit signal $PRECOM_i$ indicating its intention to participate in view $v+1$. To bind the block and pre-commit together, node $i$ computes hash $h_i = H(B_i \| PRECOM_i )$.

\textbf{PVSS Splitting. } Node $i$ now acts as a dealer for its own block. Using a PVSS scheme $\textrm{PVSS.SPLIT}(h_i, \mathcal{A}(v), t)$, ($t$ is the threshold required to reconstruct the secret $h_i$, we let $ t = \lfloor \mathcal{A}(v) / 2 \rfloor + 1$), node $i$ produces: \textbf{A public commitment} $C_i$ ensuring the validity of shares; \textbf{Encrypted shares} $Y_i = (Y_{i,1}, \ldots, Y_{i,\mathcal{A}(v)})$, where each $Y_{i,j}$ is the share sent to node $j$ from $i$.

\begin{algorithm}[t]
\caption{PVSS-BFT Protocol}
\label{alg:pvss-bft}
\begin{algorithmic}[1]
\STATE \textbf{Initialization:} Set $outputs \leftarrow \emptyset$. Initialize $V \leftarrow \emptyset$, $\textsc{NextRoundCommit} \leftarrow \emptyset$ 
\STATE \textbf{Phase 1: Block Proposal and Share Distribution}
\STATE Generate block $B_i$ and pre-commit signal $PRECOM_i$
\STATE \textbf{If} receive $\langle AWAKE \rangle$ message:
\STATE \quad Add to $\textsc{AwakeList}_i$
\STATE Compute $h'_i \leftarrow \textsc{Hash}(B_i\| PRECOM_i)$
\STATE Generate shares: $(Y_{i,1}, ..., Y_{i,n}), C_i \leftarrow \textsc{PVSS.SPLIT}(h'_i, \mathcal{A}(v), t)$
\STATE Generate VRF output: $(\rho_i, \pi_i) \leftarrow \textsc{VRF}_i(v)$
\STATE Broadcast $PROPOSE$ $\langle v, B_i, Y_{i}, C_i, \rho_i, \pi_i, PRECOM_i \rangle_i$
\STATE \textbf{Phase 2: Share Verification and Leader Election}
\STATE For each received $PROPOSE$ message:
\STATE \quad Verify VRF: $valid_{vrf} \leftarrow \textsc{VRF.VERIFY}(pk_j, \pi_j, v)$
\STATE \quad Verify share: $Y_j \leftarrow \textsc{PVSS.VERIFY}(pk_j, C_j, Y_j)$
\STATE \quad \textbf{If} $valid_{vrf}$ \textbf{and} $Y_j$: Add $j$ to valid proposal set $V$
\STATE Select leader: $L \leftarrow \arg\max_{j \in V} \rho_j$
\STATE Record $PRECOM_j$ in $\textsc{NextRoundCommit}$
\STATE Update $\textsc{AwakeList}_i$ 
\STATE Prepare $M_1 \leftarrow \langle  v, Y_{L,i} \rangle$
\STATE Prepare $M_2 \leftarrow \langle \textsc{AwakeList}_i, \textsc{NextRoundCommit}_i \rangle$
\STATE Concatenate and broadcast: $M_1 \| M_2$
\STATE \textbf{Phase 3: Secret Reconstruction and Voting}
\STATE Add nodes to $\textsc{NextRoundParticipants}$ if they appear in $> \mathcal{A}(v)/2$ $\textsc{AwakeList}$ and add nodes to $\textsc{VoteforNextRoundCommitment}$ if appearing in $> \mathcal{A}(v)/2$ $\textsc{NextRoundCommit}$s
\STATE Reconstruct secret: $secret \leftarrow \textsc{PVSS.Reconstruct}(\{Y_{L,m}\} _{m=1}^t)$
\STATE \textbf{If} $secret = \textsc{Hash}(B_l\| \textsc{precom}_l)$: $\textsc{vote}_i \leftarrow \text{true}$
\STATE Generate shares of vote: $(Y_{\textsc{vote}_i}, C_{\textsc{vote}_i}) \leftarrow \textsc{PVSS.SPLIT}(\textsc{Hash}(\textsc{vote}_i), \mathcal{A}(v), t)$
\STATE Broadcast $VOTE$ $\langle v, \textsc{vote}_i, Y_{\textsc{vote}_i}, C_{\textsc{vote}_i} \rangle_i$
\STATE \textbf{Phase 4: Confirmation and Consensus}
\STATE Upon receiving $\geq \mathcal{A}(v)/2$ valid $VOT$E messages:
\STATE \quad Broadcast $CONFIRM$ $\langle v, \textsc{Hash}(B_l\| \textsc{precom}_l) \rangle_i$
\STATE Upon collecting $\geq \mathcal{A}(v)/2$ valid CONFIRM messages:
\STATE \quad Decide on $B_l$
\end{algorithmic}
\end{algorithm}

\textbf{VRF Generation. } Node $i$ also generates a VRF output $(\rho_i, \pi_i) \leftarrow \textrm{VRF}_i(v)$, used to elect a leader in Phase 2.

Node $i$ broadcasts a $PROPOSE$ message $\langle v, B_i, Y_i, \rho_i, \pi_i, PRECOM_i \rangle_i$ to all other nodes. By including $PRECOM_i$ within $h_i$, nodes $i$ signals possible maintenance schedules. This pre-commit mechanism ensures that a reliable set of the next-round participant for PVSS is defined. Since $C_i$ is publicly recorded, $i$ cannot selectively broadcast $Y_i$ without being detected. Integrating block proposals with PVSS before the leader election ensures security by preventing equivocating or selective message broadcasting, as any inconsistency can be caught through share reconstruction and verification.

\subsection{Phase 2: Share Verification and Leader Election}
Upon receiving $PROPOSE$ messages, node $i$ verifies the validity of shares and elects a leader based on VRF values.

\textbf{Share Verification.} For each $PROPOSE$ message from node $j$, node $i$ extracts share $Y_{j,i}$ and runs $\textrm{PVSS.VERIFY}(pk_j, C_j, Y_{j,i})$ to ensure $Y_{j,i}$ is consistent with $C_j$. Node $i$ also checks $\textrm{VRF.VERIFY}(pk_j, \pi_j, v)$. If both verification succeed, node $i$ deems $j$'s message valid and include $j$ in a set $V$.

\textbf{Leader Election.} Node $i$ selects a leader $L \leftarrow \arg\max_{j\in V} \textrm{vrf}_j$. All pre-commit signals $PRECOM$ are recorded in $\text{NextRoundCommit}_i$ list to track participation intentions for round $v+1$.

If node $i$ receives an $AWAKE$ message from a newly active node $j$ during phase 1 or 2, it updates $\text{AwakeList}_i$ to include $j$. To ensure all nodes can reconstruct $L$'s block, node $i$ broadcasts $Y_{L,i}$ (i.e. the share $i$ received from the potential leader $L$), via $\langle v, Y_{\text{L},i}, \text{AwakeList}_i, \text{NextRoundCommit}_i \rangle_i$.

The early distribution of blocks and shares before leader election ensures that even if the leader becomes unavailable, other nodes can reconstruct and verify the leader's block using share reconstruction. This design prevents malicious leaders from distributing inconsistent blocks or manipulating the protocol, as block verification proceeds independently of leader presence. Since the necessary data is pre-distributed, leader unavailability after the election is inconsequential.

\subsection{Phase 3: Secret Reconstruction and Voting}
Once the leader $L$ is elected, nodes reconstruct and verify $L$'s block. This process confirms both block integrity and participant adherence to pre-commit signals.

\textbf{Secret Reconstruction.} Node $i$ collects at least $t$ distinct shares $\{Y_{L,m}\} _{m=1}^t$. Then node $i$ decrypts every collected share
$S_{L,m}=Y_{L,m}^{1/x_m}$ and reconstructs the secret using  $\textsc{PVSS.Reconstruct} ( \{S_{L,m}\} _{m=1}^t)$ and verifies it against $L$'s block hash.  If $\textrm{secret} = \textrm{Hash}(B_l \| PRECOM_l)$, $i$ sets $\textrm{vote}_i \leftarrow \textrm{true}$; otherwise $\textrm{vote}_i \leftarrow \textrm{false}$. The successful reconstruction of the leader's block not only verifies block integrity but also confirms the adherence of participants to their pre-commits, as reconstruction requires sufficient honest nodes to contribute shares.

\textbf{Voting.} Each awake node $i$ encode its vote $(Y_{i,1}^{vote}, \ldots, Y_{i,\mathcal{A}(v)}^{vote}), C_i^{vote} $ using $ \textrm{PVSS.SPLIT}(H(vote_i), \mathcal{A}(v), t) $.  

Node $i$ broadcasts a $VOTE$ message $<v, vote_i, (Y_{i,1}^{vote}, \ldots, Y_{i,\mathcal{A}(v)}^{vote}), C_i^{vote}>$.  

The protocol tracks node appearances in $\text{NextRoundCommitment}$ lists - if a node $i$'s count exceeds threshold $\mathcal{A}(v)/2$, $k$ is added to \text{VoteforNextRoundCommitment}. Similarly, if the count for awake node \( k \) exceeds the threshold, node $i$ preliminarily confirma \( k \) as awake for next round. The integration of PVSS with voting ensures that block verification is tied to pre-commit validation, as failed reconstruction indicates either malicious behavior or violation of pre-commits. 

\subsection{Phase 4: Confirmation and Consensus}

\textbf{Commit Rule.} When node $i$ receives $\geq \mathcal{A}(v)/2$ valid $VOTE$ messages for block $B_{L}$, it broadcasts a $CONFIRM$ message $\langle v, H(B_l) \rangle_i$. Upon collecting at least $\mathcal{A}(v)/2$ valid $CONFIRM$ messages, $i$ decides on $B_l$ and finalizes the participation list for round $v+1$, incorporating nodes with validated pre-commit signals $PRECOM$ and newly confirmed awake nodes.

\textbf{Error Handling.} 
When PVSS reconstruction fails or inconsistent votes are detected, the protocol initiates error correction by having all nodes broadcast their stored shares, enabling vote reconstruction to verify share-vote consistency and identify nodes whose reconstructed votes differ from their broadcast votes. This deterministic error detection approach provides cryptographic proof of both malicious behavior and commitment violations. Regarding the handling of unexpected node failures that may trigger a view change, for example, if a node receives fewer than a quorum of valid $VOTE$ or $CONFIRM$ messages, the view may be aborted. We refer readers to Section~\ref{experiment2} and \ref{Unstable Participation Analysis}, where the protocol's tolerance to node failures in each phase is rigorously analyzed.

\section{Security Analysis}
\label{security}

\subsection{Safety Analysis}
\begin{lemma}
\label{lemma1}
For any honest validators $p, q \in \mathcal{A}(v)$, if $p$ reconstructs a block $B$ from the shares of leader $L$ at time $t = 2\Delta$, then $q$ either reconstructs the same block $B$ or fails to reconstruct any block at time $t = 2\Delta$.
\end{lemma}

\begin{proof}
We consider three possible attack scenarios.

\textbf{Scenario 1} (Share forgery without modifying commits): By the \emph{public‑verifiability} property, a tuple \((pk_i,C,Y)\) passes verification if \(Y=y_i^{p(i)}\) for the unique polynomial~\(p(x)\) defined by $C$. Hence an adversary who keeps $C$ unchanged but broadcasts $Y'_{k}\neq y_k^{p(k)}$ to some receiver $k$ will inevitably fail verification.

When introducing PVSS into BFT, at the first consensus phase, every node as a dealer selects a random polynomial $p(x) = \sum_{j=0}^{t-1} \alpha_j x^j$ of degree at most $t-1$, where $\alpha_j \in \mathbb{Z}_q$ and $s = \alpha_0$. Node $i$ has a private key $x_i \in_R \mathbb{Z}_q^*$ and a public key $y_i = G^{x_i}$, where $G_q$ is a cyclic group of prime order $q$. 

Node $i$ publishes commitment $C_i = \left\{C_0,\ldots,C_{t-1}\right\} $ where $C_j=g^{\alpha_j}$.  Once $C_i$ is published, the polynomial $p(x)$ is determined. Node $i$’s correct encrypted share is \(Y_i=y_i^{p(i)}\). Every receiver $n$ computes and verifies 
\begin{align}
X_n=\prod_{j=0}^{t-1}C_j^{\,n^j}=g^{p(n)}
\end{align}

\begin{align}
\log_g X_n=\log_{y_n}Y_n
\end{align}

Suppose an adversary $M$ broadcasts an alternative
\(Y'_k\neq Y_k\) for a certain node $k$ that still passes verification to affect $k$'s reconstruction process. Equation (1) fixes \(X_k=g^{p(k)}\). For \(Y'_k\) to pass the verification, there must exist a value \(p'(k)\) such that $Y'_k = y_k^{\,p'(k)}$ and $g^{p(k)} = g^{p'(k)}$. Because $g$ has order $q$, the latter equality implies  
\(p'(k)\equiv p(k)\pmod q\).
Consequently
\(Y'_k = y_k^{\,p'(k)} = y_k^{\,p(k)} = Y_k\),
contradicting the assumption \(Y'_k\neq Y_k\).
Hence, \emph{no} adversary can produce a different share that still
verifies against the fixed commitments.

\textbf{Scenario 2} (Simultaneous forgery of commits and shares):
By $t$‑robustness, at most one secret can be reconstructed
from $C'$.  Because honest nodes accept a share only when \((C,Y_i)\) satisfies equation (2), they will never mix shares coming from two distinct commitment vectors.  Therefore, honest nodes either all stay with $C$ or all reject the forged commitment.

Specifically, if an adversary $M$ changes any commitment \(C_j\to C'_j\), she defines a new polynomial \(p'(x)\). $M$ needs to find $p'(i)$ for each $i$ such that: $g^{p'(i)} = X'_i \quad$ and $\quad y_i^{p'(i)} = Y'_i$. The challenge for $M$ is to create consistent values $X'_i$ and $Y'_i$ that satisfy the proof for each $i$, while also ensuring that these values correspond to a valid polynomial $p'(x)$ of degree $t-1$. If $M$ attempts to create $X'_i$ and $Y'_i$ that don't correspond to a valid $p'(x)$, this inconsistency will be detectable when nodes attempt to reconstruct the secret or verify the polynomial's properties.

\textbf{Scenario 3} (Leader share forgery):
Based on the synchronous network assumption, all honest validators select the same leader. Suppose a malicious node $M$ attempts to broadcast an incorrect leader share $Y'_i \neq Y_i$ during the second phase. For each node $i$, the correct share should be: $Y_i = y_i^{p(i)} = (G^{x_i})^{p(i)} = G^{x_i \cdot p(i)}$. All nodes can use the publicly published commits to compute: $Xi = \prod{j=0}^{t-1} C_j^{i^j} = g^{p(i)} $

When any node receives a share claimed to be from the leader $Y_i$, it compute $X_i = \prod_{j=0}^{t-1} C_j^{i^j}$ and Verifies $\log_g(X_i) = \log_{y_i}(Y_i)$. If $M$ broadcasts an incorrect share $Y'_i \neq Y_i$, then there must exist $p'(i) \neq p(i)$ such that: $Y'_i = y_i^{p'(i)} = G^{x_i \cdot p'(i)}$. This is impossible by the same argument as in scenario 2; hence leader share forgery is rejected immediately.

Now, consider honest validators $p, q \in \mathcal{A}(v)$ at time $t = 2\Delta$. If $p$ successfully reconstruct block $B$. $p$ must receive at least $t > n/2$ valid shares corresponding to $B$. Each share must be consistently broadcast and verified. If $q$ receives sufficient valid shares. By Scenario 1, any share from an honest validator corresponds to the same polynomial $p(x)$. By Scenario 2, malicious validators cannot forge a consistent set of shares. By Scenario 3, all honest nodes will receive the correct share. Therefore, $q$ must reconstruct the same block $B$. If $q$ receives insufficient valid shares, $q$ will fail to reconstruct any block. 
\end{proof}

\begin{lemma}
\label{lemma2}
For any honest validators $p, q \in \mathcal{A}(v)$ of view $v$, if $p$ votes for block $B$ and $q$ votes for block $B'$ at time $t = 3\Delta$, then $B = B'$.
\end{lemma}

\begin{proof}
Consider the necessary conditions for honest validators to cast votes. For an honest validator $p$ to vote for block $B$ at time $3\Delta$. $p$ must have identified the correct leader $L$ at time $2\Delta$. $p$ and $q$ must have reconstructed $B$ from leader $L$'s shares. The reconstructed block hash must match $L$'s proposed block. The reconstruction requires at least $t > n/2$ valid shares from distinct validators. 

Formally, at time \( t = 2\Delta \), all honest validators have received shares and VRF outputs from other validators. Since the leader \( L \) is agreed upon by all honest validators (due to the deterministic VRF selection and synchronous network), they all attempt to reconstruct the leader's block \( B_L \) using the PVSS shares. Honest validator \( p \) reconstructs block \( B \) from the shares of \( L \) at time \( t = 2\Delta \). Similarly, honest validator \( q \) reconstructs block \( B' \) from the shares of \( L \). Both \( p \) and \( q \) require at least \( T \) valid PVSS shares to reconstruct the block, where \( T > n/2 \). By \textbf{Lemma 1}, which states that if any honest validator reconstructs a block \( B \) from the shares of leader \( L \), then all other honest validators reconstruct the same block \( B \). Therefore, since both \( p \) and \( q \) are honest and have reconstructed blocks from \( L \)'s shares, it must be that \( B = B' = B_L \). At time \( t = 3\Delta \), honest validators \( p \) and \( q \) proceed to the voting phase. They vote for the block they have reconstructed and verified, which is \( B = B' = B_L \).
\end{proof}

\begin{lemma}
\label{lemma3}
For any view $v$, if an honest validator $p \in \mathcal{A}(v)$ receives sufficient votes ($\geq n/2$) for block $B$ at time t = $4\Delta$, then no honest validator $q \in \mathcal{A}(v)$ can receive sufficient votes for any block $B' \neq B $ at time t = $4\Delta$.
\end{lemma}

\begin{proof}
By Lemma 2, all honest validators who vote at time $3\Delta$ vote for the same block. Due to the leader election process, only the leader's block can be reconstructed and verified. Votes for any other block are invalid and will be discarded by honest validators. 

Honest validators send their votes to all validators. By synchronous network assumptions, votes at time $3\Delta$ will be received by all honest validators by time $4\Delta$. Even if Byzantine nodes send different votes to different honest nodes, all honest nodes will eventually receive all votes. For $p$ to receive sufficient votes ($\geq n/2$) for block $B$, Some honest validators must have voted for $B$. By Lemma 2, all honest validators who vote must vote for $B$. These honest votes will reach all honest validators by time $4\Delta$. Therefore, no block other than B can receive sufficient votes. Any honest validator $q$ must see the same voting result as $p$. 
\end{proof}

\begin{theorem}[Safety]
\label{safety}
If two honest nodes confirm blocks $B$ and $B'$ respectively, then $B$ does not conflict with $B'$.
\end{theorem}

\begin{proof}
Suppose two honest validators $p,q \in \mathcal{A}(v)$ decide different blocks $B$ and $B'$ at the same height in view v. For block B to be decided by honest validator $p$, $p$ must have received sufficient votes ($\geq n/2$) for $B$ and these votes must be based on successful block reconstruction at time $2\Delta$. By Lemma \ref{lemma1}, at time $2\Delta$, if any honest validator reconstructs block $B$, all other honest validators either reconstruct $B$. By Lemma \ref{lemma2}, at time $3\Delta$, all honest validators who vote must vote for the same block. By Lemma \ref{lemma3}, if $p$ receives sufficient votes for $B$ at time $4\Delta$, $q$ cannot receive sufficient votes for any $B' \neq B$. Therefore, it is impossible for $p$ and $q$ to decide on different blocks of the same height. 
\end{proof}

\subsection{Liveness Analysis}
\begin{theorem}
If a transaction \(\text{tx}\) is broadcast to all honest validators at time \( t \), then there exists a time \( t' \geq t \) such that all honest validators awake at time \( t' \) will eventually decide on a log containing \(\text{tx}\).
\end{theorem}

\begin{proof}
Assume that $\text{tx}$ is broadcast to all validators at time $t$. Due to the synchronous network assumption, all honest validators receive $\text{tx}$ by time $t + \Delta$. Starting from the view $v$ that begins at or after time $t + \Delta$, all honest validators include $\text{tx}$ in their proposed blocks. The leader election is based on VRF, which is unpredictable and uniformly random among validators. Given that the set of honest validators is greater than half of the total validators, an honest validator will eventually be selected as the leader in some future view $v' \geq v$. We proceed by considering the view $v'$ in which an honest leader $L$ is selected.

At time $t_0$, all awake validators broadcast their proposed blocks and shares. Honest leader $L$ includes $\text{tx}$ in its proposed block $B$. All honest validators receive shares and VRF values from other validators by time $t_0 + \Delta$. They agree on the leader $L$ based on the highest VRF value. At time $t_0 + 2\Delta$, honest validators reconstruct $L$'s block $B$ using shares. By Lemma \ref{lemma1}, if any honest validator reconstructs $B$, then all honest validators reconstruct the same block $B$ at time $t_0 + 2\Delta$. Since $B$ includes $\text{tx}$, all honest validators now have access to a block containing $\text{tx}$. By Lemma \ref{lemma2}, all honest validators vote for the same block $B$ at time $t_0 + 3\Delta$. Since the number of honest validators $\geq$ half of the active validator, $B$ receives sufficient votes at time $t_0 + 4\Delta$. By Lemma \ref{lemma3}, no conflicting block can receive sufficient votes. All honest validators confirm and decide on block $B$.

\textbf{Conclusion:} At time $t' = t_0 + 4\Delta$, all honest validators awake at time $t'$ have decided on a log containing $\text{tx}$. Therefore, for any transaction $\text{tx}$ broadcast at time $t$, there exists a time $t' \geq t$ such that all honest validators awake at time $t'$ will eventually decide on a log containing $\text{tx}$.
\end{proof}

\section{Experimental Evaluation}
\label{experiment}

\subsection{Experiment 1: Security Analysis}
In this experiment, we evaluate the security performance of the PVSS-BFT system against a baseline BFT protocol that omits the PVSS mechanism. By removing PVSS from our protocol, we create a conventional BFT protocol with standard security features.
The network consists of 40 nodes deployed in a cloud environment using Amazon EC2 services. We implemented both the PVSS-BFT and conventional BFT protocols on these nodes. 
We evaluate the resilience of our PVSS-BFT protocol against adversarial behaviors specifically designed to disrupt consensus and induce chain forks. 
It focuses on scenarios where malicious leaders may distribute conflicting block versions to different node groups, manipulating the consensus process to create forks. 
Specifically, if a malicious node is elected as the leader, it divides the entire set of nodes into two arbitrary groups without any coordination or knowledge of the honest nodes. The partitioning is solely determined by the malicious leader and is not based on any network topology or logical grouping. The malicious leader generates two different blocks with distinct transactions or data.
The number of malicious nodes increased from 0 to 19 to observe the systems' behavior under adversarial conditions.

% \begin{figure}[h!]
%     \centering
%     \includegraphics[width=.9\columnwidth]{image/1.png}
%     \caption{Fork occurrence comparison between baseline BFT and PVSS-BFT under malicious nodes}
%     \label{fig:image1}
% \end{figure}

% \begin{figure}[h!]
%     \centering
%     \includegraphics[width=.9\columnwidth]{image/2.png}
%     \caption{Block discard rates comparison between baseline BFT and PVSS-BFT with increasing malicious nodes}
%     \label{fig:image2}
% \end{figure}

\begin{figure}[ht]
  \centering
  \begin{minipage}[b]{0.48\textwidth}
    \centering
    \includegraphics[width=\linewidth]{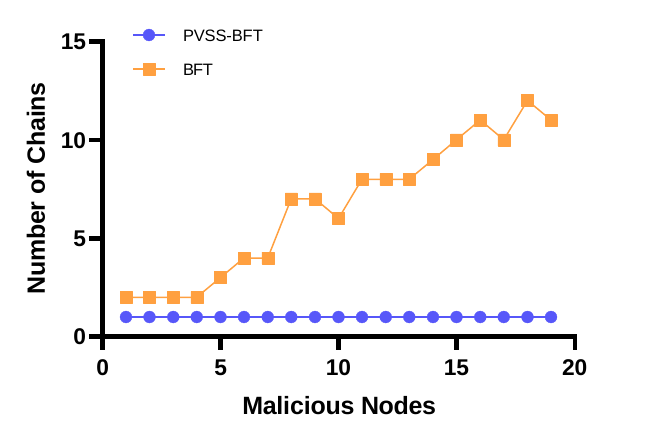}
    \caption{Fork occurrence comparison between baseline BFT and PVSS‑BFT under malicious nodes}
    \label{fig:image1}
  \end{minipage}
  \hfill
  \begin{minipage}[b]{0.48\textwidth}
    \centering
    \includegraphics[width=\linewidth]{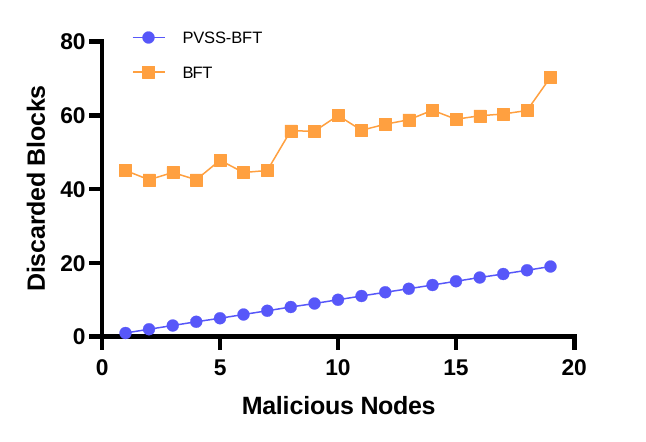}
    \caption{Block discard rates comparison between baseline BFT and PVSS‑BFT with increasing malicious nodes}
    \label{fig:image2}
  \end{minipage}
\end{figure}

\subsubsection{Forking Analysis}
Figure~\ref{fig:image1} illustrates the relationship between the number of malicious nodes and the incidence of chain forks. We show the fork resistance capabilities of both systems under increasing adversarial presence. In the baseline BFT system, we observe a clear positive correlation between the number of malicious nodes and fork occurrences. In contrast, the PVSS-BFT system shows remarkable resilience by maintaining zero forks, demonstrating its ability to cryptographically verify and reject conflicting proposals.

\subsubsection{Block Discard Analysis}
Figure~\ref{fig:image2} presents the number of blocks discarded under varying levels of adversarial presence. The baseline BFT system demonstrates a notable increase in discarded blocks as the number of malicious nodes rises, attributed to fork resolutions and conflicting proposals. The PVSS-BFT system consistently shows fewer discarded blocks, attributed to its fork prevention and rapid invalid proposal detection capabilities. 

% \begin{figure}[h!]
%     \centering
%     \includegraphics[width=0.5\columnwidth]{image/3.png}
%     \caption{Impact of malicious nodes on chain length.}
%     \label{fig:image3}
% \end{figure}

\begin{figure}[ht]
  \centering
  \begin{minipage}[b]{0.48\textwidth}
    \centering
    \includegraphics[width=\linewidth]{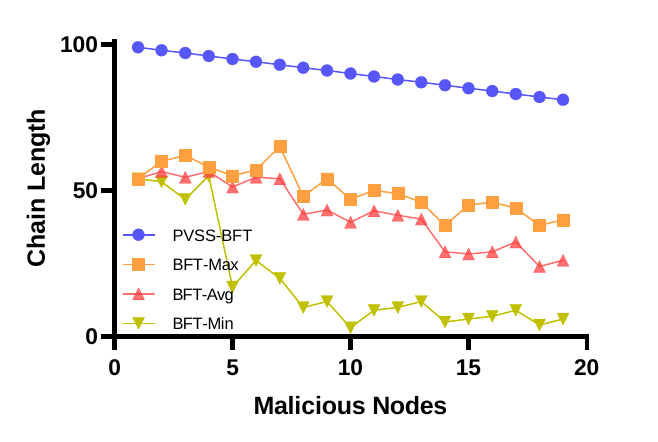}
    \caption{Impact of malicious nodes on chain length.}
    \label{fig:image3}
  \end{minipage}
  \hfill
  \begin{minipage}[b]{0.48\textwidth}
    \centering
    \includegraphics[width=\linewidth]{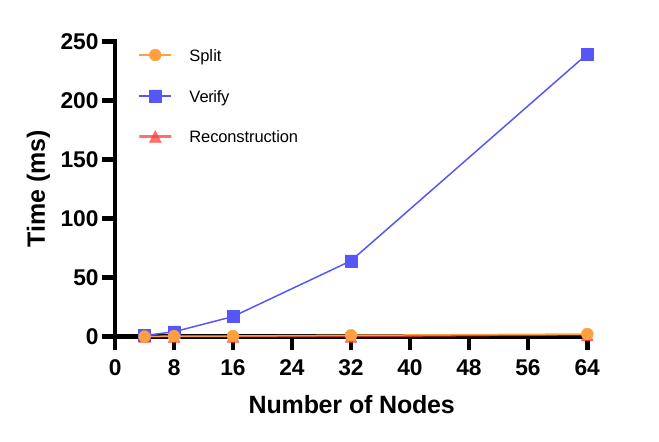}
    \caption{Time consumption of each step of PVSS.}
    \label{fig:image6}
  \end{minipage}
\end{figure}

\subsubsection{Chain Length Dynamics}
As depicted in Figure~\ref{fig:image3}, the baseline BFT system experiences significant fluctuations in chain length as the number of malicious nodes increases. This instability can be attributed to the increasing number of forks created by malicious nodes. These forks lead to frequent chain reorganizations, causing the overall chain length to fluctuate dramatically. In contrast, the PVSS-BFT system maintains a stable and consistent chain length throughout test scenarios, proving its efficacy in countering adversarial disruptions and maintaining continuous ledger growth.

% \begin{figure}[h!]
%     \centering
%     \includegraphics[width=0.9\columnwidth]{image/4.png}
%     \caption{The latency of both our protocol and the longest-chain protocol.}
%     \label{fig:latency_top}

%     \includegraphics[width=0.9\columnwidth]{image/5.png}
%     \caption{The participation level over time.}
%     \label{fig:latency_bottom}
% \end{figure}

\begin{figure}[h!]
  \centering
  \begin{minipage}[b]{0.48\columnwidth}  % 调整宽度，确保总和 <1
    \centering
    \includegraphics[width=\linewidth]{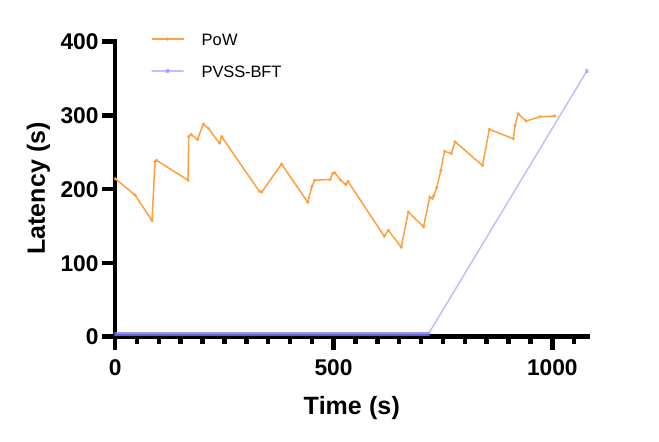}
    \caption{The latency of both our protocol and the longest-chain protocol.}
    \label{fig:latency_top}
  \end{minipage}
  \hfill  % 填充水平间距
  \begin{minipage}[b]{0.48\columnwidth}
    \centering
    \includegraphics[width=\linewidth]{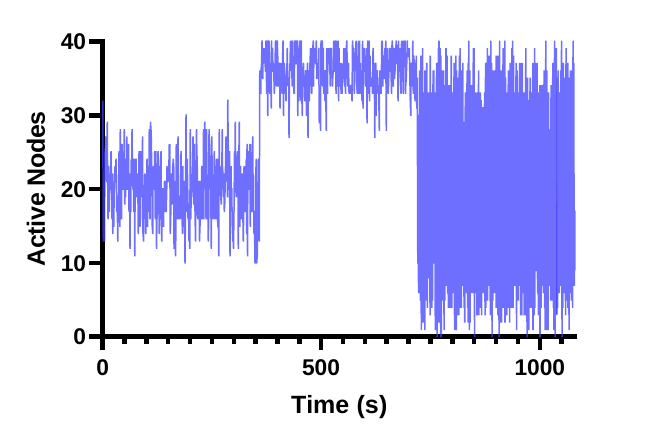}
    \caption{The participation level over time.}
    \label{fig:latency_bottom}
  \end{minipage}
\end{figure}

\subsection{Experiment 2: Latency Analysis}
\label{experiment2}
We first evaluated the PVSS's computational overhead by benchmarking the latency of its core operations: $SPLIT$, $VERIFY$, and $RECONSTRUCTION$. Figure \ref{fig:image6} summarizes the results for varying numbers of nodes ($n \in \{4, 8, 16, 32, 64\}$). At $n=64$, the average time for verification is approximately $230 ms$. The splitting and reconstruction phases take negligible time. While verification is the dominant cost, its latency remains practical. It is comparable to typical cross-region network latencies observed in different AWS servers as described in the full version of \cite{kelkar2023themis}. Crucially, given our protocol's phase time assumption of $\Delta=1s$ per phase (adopted based on the experimental setup in \cite{momose2022constant}), the computational overhead of all PVSS operations fits comfortably within this budget. Therefore, the PVSS computations do not impose a bottleneck.

We then evaluate how our protocol performs under varying participation levels, using the “Longest-Chain” approach from Momose and Ren \cite{momose2022constant} as a baseline. To the best of our knowledge, theeir experiment is the only publicly documented framework that models nodes switching online/offline at second-level granularity, while existing sleepy models (e.g., \cite{malkhi2023towards}\cite{d2023streamlining}) report performance only through theory. We emphasize that this comparison is strictly about latency under dynamic participation.
Figure~\ref{fig:latency_top} and ~\ref{fig:latency_bottom} illustrate how our protocol performs relative to the longest-chain protocol. For the longest-chain baseline, we configured the system to produce blocks approximately every 15 seconds, with a confirmation length of \( k = 10 \) blocks. Inspired by the experimental setup in~\cite{momose2022constant}, we consider three distinct participation level stages over a period of approximately 1080 time units (seconds): low participation, high participation, and unstable participation.

\begin{enumerate}
    \item Low participation (0-360 seconds): Initially, 20 nodes are active. The total active participation then fluctuates around 50\% of all nodes, following an approximate sinusoidal pattern.
    \item High participation (360-720 seconds): The participation probability was consistently set at 90\%, maintaining more than 30 active nodes. 
    \item Unstable participation (720-1080 seconds): The number of active nodes was highly volatile, fluctuating rapidly each second from 0 to 40.
\end{enumerate}

In our protocol, we observe a constant latency that is significantly lower than the longest-chain protocol. This consistent low latency is maintained throughout the first two stages. However, when participation fluctuates greatly, our protocol can ensure security, but cannot make the ledger grow.

\textbf{Unstable Participation Analysis:} In the scenario of \emph{unstable participation}, nodes can independently change their status every second with a probability \( p \). In this context, an active node can become sleepy and a sleepy node can become active with equal probability. This can lead to a situation where our protocol is unable to make progress. This challenge primarily arises from the requirement of a specific threshold of nodes for secret reconstruction. We analyze the maximum tolerable probability \( p \) that allows the protocol to function correctly without compromising safety guarantees in the appendix\ref{Unstable Participation Analysis}. Each consensus round can tolerate up to 63\% of nodes being offline.

\section{Conclusion and Future Work}
\label{conclusion}
In this paper, we have presented an advanced sleepy BFT consensus model. By integrating PVSS and VRF into a four-phase consensus mechanism, our approach offers improved efficiency and security compared to existing models. We have provided formal proof for the safety and liveness properties of our protocol, demonstrating its robustness in the face of Byzantine adversaries and dynamic participation.  Our experimental evaluations corroborate the theoretical findings. In simulated adversarial conditions, the PVSS-BFT protocol demonstrated superior fork resistance and chain stability compared to traditional BFT systems. Latency tests revealed that our protocol maintains consistently low latency across varying participation levels, only showing limitations in extreme instability scenarios.

We acknowledge that PVSS inherently adds communication overhead due to the distribution of shares and commitments. While our computational benchmarks demonstrate feasibility, optimizing communication complexity, particularly at larger network scales, remains an important consideration. Future work includes investigating techniques to reduce overall message costs.
\section*{Acknowledgment}
This project is fully supported by the CloudTech-RMIT Green Bitcoin Joint Research Program.
\bibliography{reference}

\appendix
\section{Unstable Participation Analysis}
\label{Unstable Participation Analysis}
We consider a PVSS-BFT system with \( n \) nodes. Let \( X_{i,j} \) denote the number of active nodes in phase \( j \) of round \( i \), where \( j = 1, 2, 3, 4 \) and \( S_{i,j} \) denote the number of nodes confirmed as sleepy by phase \( j \) of round \( i \). The expected number of active nodes in the first phase of round \( i \) is given by:

\begin{equation}
\label{eq:steady_state_phase1}
E[X_{i,1}] = E[X_{i-1,1}] \times (1 - p)^4 + E[Y_{i-1}],
\end{equation}

where \( E[Y_{i-1}] \) is the expected number of newly activated nodes in round \( i - 1 \) that are eligible to participate in phase 1 of round \( i \).

For subsequent phases \( j = 2, 3, 4 \), since only nodes that were active in phase 1 and remain active can participate, the expected number of active nodes is:

\begin{equation}
\label{eq:expected_active_phases}
E[X_{i,j}] = E[X_{i,1}] \times (1 - p)^{(j)}.
\end{equation}

For phase 2, the expected number of active nodes and sleepy nodes are:
\begin{equation}
\label{eq:phase2_nodes}
\begin{split}
E[X_{i,2}] &= E[X_{i,1}] \times (1 - p)^2 \\
E[S_{i,2}] &= E[X_{i,1}] \times p
\end{split}
\end{equation}

For phase 3 (vote phase), the expected number of active nodes and sleepy nodes are:
\begin{equation}
\label{eq:phase3_nodes}
\begin{split}
E[X_{i,3}] = E[X_{i,1}] \times (1 - p)^3  \\
E[S_{i,3}] = E[X_{i,1}] \times (1-(1-p)^2)
\end{split}
\end{equation}

The expected number of newly activated nodes eligible for the next round is:
\begin{equation}
\label{eq:newly_activated_nodes}
\begin{split}
    E[Y_{i}] = &\ (n - E[X_{i,1}]) \times p \times (1 - p)^3 \\
              &+ (n - E[X_{i,1}] \times (1 - p)) \times p \times (1 - p)^2.
\end{split}
\end{equation}

To reconstruct the secret and reach a consensus, the protocol must satisfy two conditions:
\begin{equation}
\label{eq:success_conditions}
\begin{split}
X_{i,3} &\geq 0.5 \times X_{i,1} \quad \text{(vote phase)} \\
X_{i,4} &\geq \frac{X_{i,1} - S_{i,2}}{2} \quad \text{(confirm phase)}
\end{split}
\end{equation}

\subsubsection{Latency Analysis Using Normal Approximation}

Assuming the system has reached a steady state, i.e., \( E[X_{i,1}] = E[X_{i-1,1}] \), we have:
\begin{equation}
\label{eq:steady_state_exi1}
E[X_{i,1}] = E[X_{i,1}] \times (1 - p)^4 + E[Y_{i-1}],
\end{equation}
\begin{equation}
\label{eq:steady_state_solution}
E[X_{i,1}] = \frac{E[Y_{i-1}]}{1 - (1 - p)^4}.
\end{equation}

Given the synchronous network assumption, we can deterministically detect node failures in phases 1-2. However, these detections must be confirmed through the voting process in phase 3. The success probability \( P_{\text{success}} \) depends on two critical conditions:
\begin{equation}
\label{eq:success_probability_vote}
\begin{split}
    P_{\text{success\_vote}} = P\left( X_{i,3} \geq 0.5 \times X_{i,1} \right)
\end{split}
\end{equation}
\begin{equation}
\label{eq:success_probability_confirm}
\begin{split}
    P_{\text{success\_confirm}} = P\left( X_{i,4} \geq \frac{X_{i,1} - S_{i,2}}{2} \,\bigg|\, X_{i,3} \geq 0.5 \times X_{i,1} \right)
\end{split}
\end{equation}

For the vote phase condition, substituting the expected values:
\begin{equation}
\label{eq:vote_phase_constraint}
E[X_{i,3}] = E[X_{i,1}] \times (1 - p)^3 \geq 0.5 \times E[X_{i,1}]
\end{equation}

This simplifies to:
\begin{equation}
\label{eq:simplified_constraint}
(1 - p)^3 \geq \frac{1}{2}
\end{equation}

Solving for the maximum tolerable probability:
\begin{equation}
\label{eq:max_p_solution}
p \leq 1 - \frac{1}{\sqrt[3]{2}} \approx 0.21
\end{equation}

The confirm phase threshold is dynamically adjusted based on detected sleepy nodes from phases 1 and 2, but this adjustment only takes effect after successful voting. For the confirm phase, given successful voting, we have:
\begin{equation}
\label{eq:confirm_phase_constraint}
E[X_{i,4}] \geq \frac{E[X_{i,1}] - E[S_{i,2}]}{2}
\end{equation}

Where $E[S_{i,2}]$ is the sum of sleepy nodes detected in phase 1 and 2:
\begin{equation}
\label{eq:sleepy_nodes_phase3}
\begin{split}
E[S_{i,4}] &= E[X_{i,1}] \times (1-(1-p)^2)
\end{split}
\end{equation}

Substituting the expressions:
\begin{equation}
\label{eq:confirm_phase_expanded}
E[X_{i,1}] \times (1 - p)^4 \geq \frac{E[X_{i,1}] - E[X_{i,1}] \times (1-(1-p)^2)}{2}
\end{equation}

Let $p = 0.21$ (derived from vote phase constraint). Substituting this value into the confirm phase inequality shows that it is satisfied. Therefore, the vote phase constraint $p \leq 1 - \frac{1}{\sqrt[3]{2}} \approx 0.21$ is indeed the bottleneck of our protocol. Each consensus round can tolerate up to \( 1 - (1 - 0.21)^4 \approx 0.63 \) or 63\% of nodes being offline at some point during the round. This significant improvement over traditional fixed-threshold approaches is achieved through the combination of deterministic sleepy node detection in the synchronous network and dynamic threshold adjustment after successful voting.

\end{document}